\documentclass{article}

\usepackage{tikz, pgf, pgfplots}
\usetikzlibrary
	{
		arrows, automata, chains, datavisualization.formats.functions, fit, positioning, quantikz, shapes
	}

\usepackage{amsthm}
\usepackage{amssymb}
\usepackage[bb = boondox]{mathalfa}
\usepackage{dsfont}
\usepackage{mathtools}
\usepackage{multicol}
\usepackage{physics}
\numberwithin{equation}{section}

\usepackage{yquant, braket}
\usetikzlibrary{quantikz, fit}

\usepackage[a4paper, left = 2.5cm, right = 2.5 cm, top = 2.5cm, bottom = 3.0 cm]{geometry}

\usepackage{xcolor}
\usepackage{varwidth}
\usepackage[breakable, skins, theorems]{tcolorbox}
\usepackage{graphicx}
\usepackage{hyperref}
\hypersetup
{
	colorlinks,
	linkcolor = {WordRed!75!black},
	citecolor = {WordBlueDarker25},
	urlcolor = {WordAquaDarker50}
}
\usepackage{nicematrix}
\usepackage{colortbl}
\usepackage{float}
\usepackage{cellspace}
\usepackage{makecell}
\usepackage[boxed, ruled, vlined]{algorithm2e}
\usepackage{appendix}

\newtheorem{definition}{Definition}[section]
\newtheorem{theorem}{Theorem}[section]

\newtheorem{proposition}[theorem]{Proposition}
\newtheorem{corollary}[theorem]{Corollary}
\newtheorem{example}{Example}[section]

\definecolor{MyLightRed}{RGB}{244, 213, 245}
\definecolor{WordRed}{RGB}{255, 0, 102}
\definecolor{RedDarkLightest}{HTML}{ff0088}
\definecolor{RedDarkLight}{HTML}{ea005f}
\definecolor{RedPurple}{HTML}{aa007f}
\definecolor{Purple}{HTML}{911146}
\definecolor{WordLightGreen}{RGB}{140, 214, 192}
\definecolor{WordGreen}{RGB}{0, 176, 80}
\definecolor{GreenLightest}{HTML}{00ffa0}
\definecolor{GreenLighter1}{HTML}{00b383}
\definecolor{GreenLighter2}{HTML}{00aa7f}
\definecolor{GreenDark}{HTML}{225522}
\definecolor{GreenTeal}{HTML}{008080}
\definecolor{WordIceBlue}{RGB}{223, 227, 229}
\definecolor{MyVeryLightBlue}{RGB}{211, 245, 247}
\definecolor{WordBlueVeryLight}{RGB}{0, 176, 240}
\definecolor{WordBlueLight}{RGB}{0, 112, 192}
\definecolor{WordBlueDark}{RGB}{46, 116, 181}
\definecolor{WordBlueDarker}{RGB}{31, 78, 121}
\definecolor{WordBlueDarker25}{RGB}{54, 96, 146}
\definecolor{WordBlueDarker50}{RGB}{36, 64, 98}
\definecolor{WordBlueDarkest}{RGB}{0, 32, 96}
\definecolor{WordBlue}{RGB}{19, 65, 99}
\definecolor{MyBlue}{RGB}{0, 64, 128}
\definecolor{MyDarkBlue}{RGB}{0, 51, 102}
\definecolor{BlueVeryDark}{HTML}{222255}
\definecolor{WordAquaLighter80}{RGB}{218, 238, 243}
\definecolor{WordAquaLighter60}{RGB}{183, 222, 232}
\definecolor{WordAquaLighter40}{RGB}{146, 205, 220}
\definecolor{WordAquaDarker25}{RGB}{49, 134, 155}
\definecolor{WordAquaDarker50}{RGB}{33, 89, 103}
\definecolor{WordVeryLightTeal}{RGB}{223, 236, 235}
\definecolor{WordLightTeal}{RGB}{160, 199, 197}
\definecolor{WordDarkTealLighter80}{RGB}{207, 223, 234}
\definecolor{WordDarkTeal}{RGB}{72, 123, 119}
\definecolor{WordDarkerTeal}{RGB}{48, 82, 80}
\definecolor{WordTurquoiseLighter80}{RGB}{209, 238, 249}

\title{Conditions that enable a player to surely win in sequential quantum games}

\author{
	Theodore Andronikos$^1$\\
	$^1$Ionian University, Department of Informatics,\\
	7 Tsirigoti Square, 49100 Corfu, Greece;\\
	andronikos@ionio.gr \\
}

\begin{document}

\maketitle

\begin{abstract}
	This paper studies sequential quantum games under the assumption that the moves of the players are drawn from groups and not just plain sets. The extra group structure makes possible to easily derive some very general results characterizing this class of games. The main conclusion of this paper is that the specific rules of a game are absolutely critical. The slightest variation in the rules may have important impact on the outcome of the game. This work demonstrates that it is the combination of two factors that determines who wins: (i) the sets of admissible moves for each player, and (ii) the order of moves, i.e., whether the same player makes the first and the last move. Quantum strategies do not a priori prevail over classical strategies. By carefully designing the rules of the game it is equally feasible either to guarantee the fairness of the game, or to give the advantage to either player.

	\textbf{Keywords:} Quantum games; classical games; sequential quantum games; groups of actions; winning strategies.
\end{abstract}

\section{Introduction} \label{sec:Introduction}

Game theory is a powerful tool that can be especially helpful whenever one wants to study complex situations, concerning conflict, competition, cooperation, and, even, survival between rational entities, typically refereed to as players. Despite its rather playful terminology and setting, game theory has been used to address serious problems from a multitude of fields, ranging from economics, political and social sciences, to computer science, evolutionary biology and psychology. It is not at all surprising that for such a popular subjects there are many superb textbooks. Such a modern and pedagogical introduction is the book \cite{Dixit2015}. For the reader striving for more rigor, the book \cite{Myerson1997} by Myerson and the book by Maschler et al. \cite{Maschler2020} fit the bill. These books are the main references for the game-theoretic concepts used in this paper. Game theory provides an elaborate framework that facilitates the analysis of conflict, competition, and cooperation. This analysis is primarily based on the concept of players or agents. Every agent has her own agenda and her actions interfere with the plans of the other agents. When the game ends, every player evaluates her gains or losses using a payoff function. The assumption of rationality is essential in order to ensure that every agent strives to maximize her payoff.

The field of quantum computation promises to bring a new era in our computational abilities. Although its origins can be traced back to the early '80s, it took a while for the field to mature and reach its current state. Today there exist commercial quantum computers utilizing a rather modest number number of qubits. This is expected to change in the coming years; if the predictions come true, the number of available qubits will increase substantially. This, in turn, will offer a spectacular increase in our computational capabilities and enable us to tackle useful and practical instances of difficult decision and optimization problems. In this larger perspective it would seem that the emergence of quantum games was unavoidable. The recent field of quantum game theory studies the use of quantum ideas and concepts in classical games, such as the coin flipping, the prisoners' dilemma and many others.

\subsection{Related work}

The field of quantum games was initiated in 1999 when Meyer in his influential paper \cite{Meyer1999} presented the PQ penny flip game, which is the quantum version of the classical penny flip game. One of the major attractions of his formulation was the inclusion of two well-known fictional characters, Picard and Q, from the famous tv series Star Trek, that compete by acting upon a quantum coin. Picard represents the classical player and Q the quantum player. This means that Picard behaves as a player in the classical penny flip game, but Q can use arbitrary unitary operator, a hallmark of the quantum domain. Mayer proved that in this game Q can always win using the venerable Hadamard transform. Later, many researchers motivated by Meyer's ideas, generalized the same game pattern to $n$-dimensional quantum systems. This line of research was pursued in \cite{Wang2000}, \cite{Ren2007}, and \cite{Salimi2009}. Later, a complete correspondence of every finite variant of the PQ penny flip game with easily constructed finite automata was given in \cite{Andronikos2018}. Recently, the same game was extensively analyzed from the viewpoint of groups and its relation to the dihedral groups was established in \cite{Andronikos2021}. In many works the quantum player seems to have a clear advantage over the classical player. However, things are not that black and white, as was shown in \cite{Anand2015}. The authors there cleverly changed the rules of the PQ penny flip game and this allowed the classical player to win. Another general and interesting problem, that of quantum gambling based on Nash equilibrium, was also recently studied in \cite{Zhang2017}. In a similar vein, quantum coin flipping has ben used as a crucial ingredient of many quantum cryptographic protocols, where the usual suspects Alice and Bob assume the role of remote parties that have to agree on a random bit (see \cite{Bennett2014} for more details). This idea has been extended in \cite{Aharon2010} to quantum dice rolling in settings of multiple outcomes and more than two parties.

A particularly influential work, also from 1999, was the one by Eisert et al. in \cite{Eisert1999}. They introduced a novel technique, the Eisert-Wilkens-Lewenstein protocol, which is now extensively used in the literature. Using their technique, they defined a variant of the famous prisoners' dilemma and demonstrated that there exists a quantum strategy that is better than any classical strategy. Subsequently, many researchers followed that line of research, obtaining interesting results. For some recent developments one may consult \cite{Giannakis2019}, where the correspondence of typical conditional strategies used in the classical repeated prisoners' dilemma game to languages accepted by quantum automata was established, and \cite{Rycerz2020}, where the Eisert–Wilkens–Lewenstein protocol was generalized.

The connection between game theory, and in particular infinitely repeated games, and finite automata was initially investigated in the works of \cite{Neyman1985}, \cite{Rubinstein1986}, \cite{Abreu1988}, and \cite{Marks1990}. Subsequently, Meyer used quantum lattice gas automata to study Parrondo games in \cite{Meyer2002}. The application of probabilistic automata in the prisoners' dilemma game was undertaken in \cite{Bertelle2002}, while \cite{Suwais2014} considered the use of different automata variants in game theory. In \cite{Giannakis2015a} it was pointed out that quantum automata operating on infinite words can encode winning strategies for abstract quantum games.

It is worth noting that the idea to resort to unconventional means to achieve better results in classical games is not limited to the quantum domain. The most famous classical games, such as the prisoners' dilemma, have been cast in terms of biological notions (see \cite{Kastampolidou2020} for a survey). Actually, a lot of classical games can be expressed in the context of biological and bio-inspired processes (see \cite{Kastampolidou2020,Theocharopoulou2019,Kastampolidou2020a} for more references).

\subsection{Contribution}

This paper studies sequential quantum games, under the premise that the sets of moves available to the players are groups of unitary operators and not just plain vanilla sets. The additional structure present in the groups facilitates the derivation of quite general results that completely clarify the existence or not of winning strategies. Emphasis is placed not only in the case where both players are on an equal footing in terms of available actions and number of moves, but also to the case where one player has a strong advantage over the one. This advantage may concern either the number of moves, or the set of available actions, or both. It is established that who wins depends on the combination of two factors: (i) the sets of admissible moves for each player, and (ii) the order of moves, i.e., whether the same player makes the first and the last move. An important conclusion is that quantum strategies do not necessarily prevail over classical ones. The game designer has the ability to either guarantee the fairness of the game, or to give a decisive advantage to one of the players. The results obtained in this paper are, to the best of our knowledge, stated and proved in this generality for the first time in the literature.

\subsection{Organization}

This paper is organized as follows: Section \ref{sec:Introduction} provides the most relevant references, while Section \ref{sec:Background} explains the concepts and notation used in this paper. Sections \ref{sec:Results for Canonical Games} and \ref{sec:Results for Asymmetric Games} state and prove the main results of this work. Section \ref{sec:Discussion and Examples} presents two illustrative examples and discusses the most important findings. Finally, Section \ref{sec:Conclusions} contains the main conclusions, along with some pointers for future work.

\section{Background} \label{sec:Background}

This paper studies \emph{sequential} \emph{zero-sum} quantum games that are played between two rational players, namely player 1 and player 2. The term sequential means that there is a specific predefined order of play. The two players do not act simultaneously, but take turns making their moves, always adhering to the predefined order. A single move by any player will be referred to as a \emph{round}.
Of course, zero-sum implies that the amount of loss of player 1 is equal to the amount of gain of player 2 and vice versa. Before we proceed any further, let us formally state the assumptions on which our analysis is based.

\begin{enumerate}
	\item[\textbf{A1:}]	The object of the game is a fixed $n$-dimensional quantum system. This can be any finite dimensional quantum system, but if a more concrete visualization can enhance our perception, we may think of a quantum coin, or a quantum dice, or a quantum roulette.
	\item[\textbf{A2:}]	The game takes place in the $n$-dimensional complex Hilbert space $\mathcal{H}_{n}$. The computational basis, $H_{n}$, consists of the $n$ basis states: $\ket{0},$ $\ket{1},$ $\dots,$ $\ket{n - 1}$.
	\item[\textbf{A3:}]	The system is initially in one of its basis states. This is predefined and agreed upon by both players. Henceforth, we shall refer to this basis state as the \emph{initial} state and designate it by $\ket{q_{0}} \in H_{n}$.
	\item[\textbf{A4:}]	The two players take turns acting upon the system according to some predefined order.
	$A$ and $B$ denote the sets of admissible moves of players 1 and 2, respectively. The two players draw their moves from their respective set when acting upon the system. In this paper, we shall assume that $A$ and $B$ are nontrivial \emph{subgroups} of $U(n)$, the unitary group of degree $n$. We use the standard notation $A, B \leq U(n)$ to denote this fact, and we refer to them as \emph{action groups}.
	\item[\textbf{A5:}]	Both players know the initial state of the system, but after the game begins and while it unfolds the state of the system is unknown to both of them. After the game ends, the state of the system is measured in the computational basis. If it is found to be $\ket{q_{A}} \in H_{n}$ then player 1 wins, whereas if it is found to be $\ket{q_{B}} \in H_{n}$ then player 2 wins. We shall assume that $\ket{q_{A}}$ and $\ket{q_{B}}$ are \emph{different}, and refer to $\ket{q_{A}}$ and $\ket{q_{B}}$ as the \emph{target} state of player 1 and the \emph{target} state of player 2, respectively.
	\item[\textbf{A6:}]	Therefore, each player knows:
	\begin{enumerate}
		\item	the set of moves of the other player,
		\item	the initial state $\ket{q_{0}}$, and
		\item	the target state of the other player.
	\end{enumerate}
	The players do not know the state of system after the game begins, and until they make the final measurement.
\end{enumerate}

So, the situation in a quantum game is in stark contrast to the typical situation of classical games, where both players are aware of the state of the system as the game progresses, and are thus able, at least in principle, to adapt their strategies. The fact that none of the players knows the state of the system deprives them of critical information that might have helped them to win the game. In effect it places an additional difficulty upon the decision-making process of the players.

At this point, we make the following helpful remarks:
\begin{enumerate}
	\item	More general situations can be easily envisioned. For example, the initial state may be a (maximally) entangled state or the target state (for either player) might also be entangled (see also the important reference \cite{Anand2015}). Nevertheless, we shall not concern ourselves with such cases here.
	\item	It is certainly possible to think of games where the sets of moves $A$ and $B$ are just (finite) subsets of $U(n)$, i.e., without the additional structure of a group. However, by assuming that $A$ and $B$ are groups, we can easily arrive at interesting conclusions. The concepts and the notation from group theory will be kept to a minimum and, in any case, can be found in standard textbooks such as \cite{Artin2011} and \cite{Dummit2004}. For unitary groups in particular more specialized text such as \cite{Stillwell2008} and \cite{Hall2013} can be consulted.
\end{enumerate}

In the context of sequential quantum games, it is convenient to employ the terminology outlined in Definition \ref{def:Strong and Weak Winning Strategies}, adapted from \cite{Maschler2020} and \cite{Myerson1997} to fit our analysis. Informally, the word \emph{strategy} implies the existence of a plan on behalf of each player. This plan is composed of actions, or moves that the player makes as the game evolves.

\begin{definition} [Strong and weak winning strategies] \label{def:Strong and Weak Winning Strategies} \
	\begin{itemize}
		\item	A \emph{strategy} $\sigma_{1} = (A_1, \dots, A_{m_1})$ $( \sigma_{2} = (B_1, \dots, B_{m_2}) )$ for player 1 $($player 2$)$ is a sequence of admissible moves, where $m_1$ $(m_2)$ is the number of rounds that player 1 $($player 2$)$ acts on the system, and $A_1, \dots, A_{m_1} \in A$ $( B_1, \dots, B_{m_2} \in B )$.
		\item	A strategy $\sigma_{1}$ $( \sigma_{2} )$ is a \emph{strong winning} strategy for player 1 $($player 2$)$ if for \emph{every} strategy $\sigma_{2}$ $( \sigma_{1} )$ of player 2 $($player 1$)$, player 1 $($player 2$)$ wins the game with probability $1.0$.
		\item	A strategy $\sigma_{1}$ $( \sigma_{2} )$ is a \emph{weak winning} strategy for player 1 $($player 2$)$ if for \emph{there exists} a strategy $\sigma_{2}$ $( \sigma_{1} )$ of player 2 $($player 1$)$, such that player 1 $($player 2$)$ wins the game with probability $1.0$.
	\end{itemize}
\end{definition}

Since the main concern of this work is to examine when a player can win with probability $1.0$ no matter what the other player does, it focuses on strong winning strategies. This situation is more succinctly described by saying that the player \emph{surely wins} the game. The notion of weak winning strategy can serve simply as an indicator of a possible win, in contrast to a certain win.

It is conceivable that a game may be devised so that one of the players or even both players act on the system for a number of consecutive rounds. For instance, perhaps in a game one player makes the moves $A_{i}, A_{i + 1}, \dots, A_{j}$ back-to-back, without the other player being allowed to make an intermediate move. Such a situation can be greatly simplified due to the fact that the moves available to each player are elements of a group. This means that their composition is also an element of the same group and as such it is available to the player. In the preceding example, the action $C = A_{j} \dots A_{i + 1} A_{j}$, is also an admissible move for the player. Hence, for the rest of this paper, the following useful simplification is adopted: \emph{no player plays two or more successive rounds and the two players alternate each round}. This motivates the next Definition \ref{def:Canonical and Noncanonical Games}.

\begin{definition} [Canonical and noncanonical games] \label{def:Canonical and Noncanonical Games} \
	\begin{itemize}
		\item	A game is called \emph{canonical} if its duration is $2 m$ rounds, player 1 plays at odd rounds $1, 3, \dots, 2 m - 1$, and player 2 at even rounds $2, 4, \dots, 2 m$. In a canonical game both players play the same number of rounds taking turns making their moves, player 1 plays first, and player 2 plays last.
		\item	A game is called \emph{noncanonical} if it has $2 m + 1$ rounds, player 1 plays at odd rounds $1, 3, \dots, 2 m - 1, 2 m + 1$, and player 2 at even rounds $2, 4, \dots, 2 m$. In a noncanonical game both players take turns making their moves, but player 1  plays one more round than player 2 and makes the \emph{first} and the \emph{last} move.
	\end{itemize}
\end{definition}

Let us remark that it is the assumption that the repertoire of actions for each player constitutes a group that makes this simplification possible. If the player's set of moves is just a set and not a group, then it is quite possible that the composition of moves is not an element of this set.

It would not make any difference if in the definition of canonical games it was player 2 that made the first move. The important point is that canonical games capture the notions of symmetry and fairness, with respect to sequential quantum games, in the sense that both players have the exact same number of moves, one gets to act first and the other last. Similarly, in noncanonical games it is not substantial that player 1 makes the first and the last move. It would be equally viable to be player 2. What matters is that noncanonical games model asymmetric situations when one of the players seems to have the advantage of both starting and finishing the game. An abstract visual representation of the form of canonical and noncanonical games is given in Figures \ref{fig:Canonical Game} and \ref{fig:Noncanonical Game}, respectively. 

\vspace{0.5 cm}

\begin{figure}[H]
	\centering
	\begin{tikzpicture} [scale = 1.0]
		\begin{yquant}
			qubits {$\ket{q_0}$} GAME;
			hspace {0.2 cm} GAME;
			box {$A_1$} GAME;
			hspace {0.2 cm} GAME;
			box {$B_1$} GAME;
			hspace {0.2 cm} GAME;
			box {$A_2$} GAME;
			hspace {0.2 cm} GAME;
			box {$B_2$} GAME;
			hspace {0.2 cm} GAME;
			[draw = none]
			box {$\dots$} GAME;
			hspace {0.2 cm} GAME;
			box {$A_{m}$} GAME;
			hspace {0.2 cm} GAME;
			box {$B_{m}$} GAME;
			hspace {0.2 cm} GAME;
			measure GAME;
			hspace {0.2 cm} GAME;
			output {$\ket{q?}$} GAME;
		\end{yquant}
	\end{tikzpicture}
	\caption{This figures gives a schematic representation of a $2 m$ round canonical game in which the two players take turns playing $m$ round each. Player 1 makes the first move and player 2 the last. The question mark ``?'' is used to convey the fact that the outcome of the measurement cannot be determined with probability 1.0 beforehand.}
	\label{fig:Canonical Game}
\end{figure}

\vspace{0.3 cm}

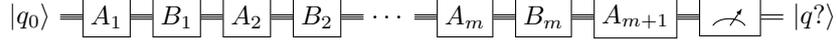
\begin{figure}[H]
	\centering
	\begin{tikzpicture} [scale = 1.0]
		\begin{yquant}
			qubits {$\ket{q_0}$} GAME;
			hspace {0.2 cm} GAME;
			box {$A_1$} GAME;
			hspace {0.2 cm} GAME;
			box {$B_1$} GAME;
			hspace {0.2 cm} GAME;
			box {$A_2$} GAME;
			hspace {0.2 cm} GAME;
			box {$B_2$} GAME;
			hspace {0.2 cm} GAME;
			[draw = none]
			box {$\dots$} GAME;
			hspace {0.2 cm} GAME;
			box {$A_{m}$} GAME;
			hspace {0.2 cm} GAME;
			box {$B_{m}$} GAME;
			hspace {0.2 cm} GAME;
			box {$A_{m + 1}$} GAME;
			hspace {0.2 cm} GAME;
			measure GAME;
			hspace {0.2 cm} GAME;
			output {$\ket{q?}$} GAME;
		\end{yquant}
	\end{tikzpicture}
	\caption{This figures shows a $2 m + 1$ round noncanonical game in which the player 1 plays $m + 1$ rounds but player 2 only $m$ rounds. Player 1 makes the first and the last move. In this case too the question mark ``?'' is used to convey the fact that the outcome of the measurement cannot be determined with probability 1.0 beforehand.}
	\label{fig:Noncanonical Game}
\end{figure}

\section{Results for canonical games} \label{sec:Results for Canonical Games}

In the subsequent study one well-known classical group will feature prominently. This is the group of all permutations of the indices $\{1, 2, \dots, n\}$, which is called the \emph{symmetric group}, and is denoted by $S_n$. If one distinguishes between the quantum and the classical player, as is often done in this area, the symmetric group in a sense encodes the maximal capabilities of the classical player. Thus it makes sense to consider it as a point of reference when analyzing how the classical player can act on the system. The elements of the symmetric group have a standard matrix representation (see also \cite{Artin2011} and \cite{Meyer2000}) in the form of permutation matrices. A \emph{permutation matrix} has a single 1 in each row and in each column and all its other entries are 0.

The identity operator, which is the identity element of the groups $S_{n}$ and $U(n)$, is denoted by $I$ and is succinctly written as $I = \begin{bmatrix} \ket{0} | \dots | \ket{i} | \dots | \ket{j} | \dots | \ket{n - 1} \end{bmatrix}$. Among permutation matrices, especially useful are the \emph{transposition} matrices. Transposition matrices are designated by $T_{i, j}$, where $0 \leq i \neq j \leq n - 1$.

\begin{align} \label{eq:Transposition Matrix Definition}
	T_{i, j}
	=
	\NiceMatrixOptions
	{
		code-for-first-row = \color{WordRed},
	}
	\begin{bNiceMatrix}[first-row]
		0^{th}	&	\dots	&	i^{th}	&	\dots	&	j^{th}	&	\dots	&	(n - 1)^{th}	\\
		\ket{0}	&	\dots	&	\ket{j}	&	\dots	&	\ket{i}	&	\dots	&	\ket{n - 1}
	\end{bNiceMatrix}
	\ .
\end{align}

Transposition matrices can interchange basis kets as follows:

\begin{align} \label{eq:Transposition Matrix Properties}
	T_{i, j} \ket{i} = \ket{j}
	\qquad \text{and} \qquad
	T_{i, j} \ket{j} = \ket{i}
	\ .
\end{align}

\begin{example}
For instance, when $n = 7$, the computational basis $H_{7}$ consists of the $7$ basis states: $\ket{0}, \ket{1}, \dots, \ket{6}$ shown below, and the corresponding symmetric group is $S_7$.

\begin{align} \label{eq:The Computational Basis $H_{7}$}
	\ket{0} =
	\begin{bNiceMatrix}
		1 \\
		0 \\
		0 \\
		0 \\
		0 \\
		0 \\
		0
	\end{bNiceMatrix}, \
	\ket{1} =
	\begin{bNiceMatrix}
		0 \\
		1 \\
		0 \\
		0 \\
		0 \\
		0 \\
		0
	\end{bNiceMatrix}, \
	\ket{2} =
	\begin{bNiceMatrix}
		0 \\
		0 \\
		1 \\
		0 \\
		0 \\
		0 \\
		0
	\end{bNiceMatrix}, \
	\ket{3} =
	\begin{bNiceMatrix}
		0 \\
		0 \\
		0 \\
		1 \\
		0 \\
		0 \\
		0
	\end{bNiceMatrix}, \
	\ket{4} =
	\begin{bNiceMatrix}
		0 \\
		0 \\
		0 \\
		0 \\
		1 \\
		0 \\
		0
	\end{bNiceMatrix}, \
	\ket{5} =
	\begin{bNiceMatrix}
		0 \\
		0 \\
		0 \\
		0 \\
		0 \\
		1 \\
		0
	\end{bNiceMatrix}, \
	\ket{6} =
	\begin{bNiceMatrix}
		0 \\
		0 \\
		0 \\
		0 \\
		0 \\
		0 \\
		1
	\end{bNiceMatrix} \ .
\end{align}

As a typical example of a transposition matrix we may consider $T_{3, 5}$.

\begin{align} \label{eq:Transposition Matrix $T_{3, 5}$}
	T_{3, 5}
	\overset{ (\ref{eq:Transposition Matrix Definition}) }{=}
	\begin{bNiceMatrix}
		\ket{0} & \ket{1} & \ket{2} & \ket{5} & \ket{4} & \ket{3} & \ket{6}
	\end{bNiceMatrix}
	=
	\begin{bNiceMatrix}
		1 & 0 & 0 & 0 & 0 & 0 & 0 \\
		0 & 1 & 0 & 0 & 0 & 0 & 0 \\
		0 & 0 & 1 & 0 & 0 & 0 & 0 \\
		0 & 0 & 0 & 0 & 0 & 1 & 0 \\
		0 & 0 & 0 & 0 & 1 & 0 & 0 \\
		0 & 0 & 0 & 1 & 0 & 0 & 0 \\
		0 & 0 & 0 & 0 & 0 & 0 & 1
	\end{bNiceMatrix} \ .
\end{align}

It is easy to verify that the properties outlined in (\ref{eq:Transposition Matrix Properties}) are satisfied:

\begin{align} \label{eq:Transposition Matrix $T_{3, 5}$ Properties}
	T_{3, 5} \ket{3} = \ket{5}
	\qquad \text{and} \qquad
	T_{3, 5} \ket{5} = \ket{3}
	\ .
\end{align}

\end{example}

The following fact, which is self-evident, is recorded as Proposition \ref{prp:Condition for Strong Winning Strategy} for future reference.

\begin{proposition} \label{prp:Condition for Strong Winning Strategy}
	In any type of game, canonical or noncanonical, player 1 $($player 2$)$ has a strong winning strategy, if and only if the state of the quantum system $\ket{\psi}$ prior to measurement is the target state $\ket{q_{A}}$ $(\ket{q_{B}})$ of player 1 $($player 2$)$.
\end{proposition}

\begin{theorem}[The same action group for both players] \label{thr:CSG Same action group for Both Players}
	Consider a canonical game where both players use the same action group. If the action group contains $S_{n}$ as a subgroup, i.e., $S_{n} \leq A = B \leq U(n)$, no player has a strong winning strategy, but both players have a weak winning strategy.
\end{theorem}
\begin{proof}[Proof of Theorem \ref{thr:CSG Same action group for Both Players}]
	Suppose to the contrary that player 1 has a strong winning strategy $\sigma_{1}$ $=$ $(A_1,$ $A_2,$ $\dots,$ $A_m)$, where $A_1, A_2, \dots, A_m \in A$. This means that for every strategy $\sigma_{2}$ of player 2, player 1 will surely win. 

	Consider the following two strategies for player 2: $\sigma_{2} = (B_1, B_2, \dots, B_{m - 1}, I)$ and $\sigma'_{2}$ $=$ $(B_1,$ $B_2,$ $\dots,$ $B_{m - 1},$ $T_{q_A, q_B})$, where $B_1, B_2, \dots, B_{m - 1}, I, T_{q_A, q_B} \in B$ because $B$ contains $S_{n}$. In view of Proposition \ref{prp:Condition for Strong Winning Strategy}, it must hold that
	\begin{align}
		I A_m \dots B_1 A_1 \ket{q_{0}} &= \ket{q_{A}} \label{eq:CSG Strong Winning Strategy Condition for Player 1 - I} \qquad \text{and}
		\\
		T_{q_A, q_B} A_m \dots B_1 A_1 \ket{q_{0}} &= \ket{q_{A}} \label{eq:CSG Strong Winning Strategy Condition for Player 1 - II} \ .
	\end{align}
	Equation (\ref{eq:CSG Strong Winning Strategy Condition for Player 1 - I}) implies that
	\begin{align}
		A_m \dots B_1 A_1 \ket{q_{0}} = \ket{q_{A}} \label{eq:CSG Strong Winning Strategy Condition for Player 1 - III} \ .
	\end{align}
	But then
	\begin{align} \label{eq:CSG Strong Winning Strategy Contradiction for Player 1}
		T_{q_A, q_B} A_m \dots B_1 A_1 \ket{q_{0}}
		\overset
		{ ( \ref{eq:CSG Strong Winning Strategy Condition for Player 1 - III} ) }
		{ = }
		T_{q_A, q_B} \ket{q_{A}} = \ket{q_{B}} \ ,
	\end{align}
	which contradicts (\ref{eq:CSG Strong Winning Strategy Condition for Player 1 - II}). Thus, player 1 does not have a strong winning strategy.

	Symmetrically, assume again to the contrary that player 2 has a strong winning strategy $\sigma_{2} = (B_1, B_2, \dots, B_{m - 1}, B_{m})$, where $B_1, B_2, \dots, B_m \in B$. Consider then the following two strategies for player 1: $\sigma_{1} = (A_1, A_2, \dots, A_{m - 1}, B_{m}^{-1})$ and $\sigma'_{1} = (A_1, A_2, \dots, A_{m - 1},$ $B_{m}^{-1} T_{q_A, q_B})$, where $A_1,$ $A_2,$ $\dots,$ $A_{m - 1},$ $B_{m}^{-1},$ $B_{m}^{-1} T_{q_A, q_B}$ $\in$ $A$ because $A = B$ and contains $S_{n}$. From Proposition \ref{prp:Condition for Strong Winning Strategy} we know that
	\begin{align}
		B_{m} B_{m}^{-1} B_{m - 1} A_{m - 1} \dots B_1 A_1 \ket{q_{0}} &= \ket{q_{B}} \label{eq:CSG Strong Winning Strategy Condition for Player 2 - I} \qquad \text{and}
		\\
		B_{m} B_{m}^{-1} T_{q_A, q_B} B_{m - 1} A_{m - 1} \dots B_1 A_1 \ket{q_{0}} &= \ket{q_{B}} \label{eq:CSG Strong Winning Strategy Condition for Player 2 - II} \ .
	\end{align}
	Equation (\ref{eq:CSG Strong Winning Strategy Condition for Player 2 - I}) implies that
	\begin{align}
		B_{m - 1} A_{m - 1} \dots B_1 A_1 \ket{q_{0}} &= \ket{q_{B}} \ , \label{eq:CSG Strong Winning Strategy Condition for Player 2 - III}
	\end{align}
	which together with (\ref{eq:CSG Strong Winning Strategy Condition for Player 2 - II}) gives
	\begin{align} \label{eq:CSG Strong Winning Strategy Contradiction for Player 2}
		B_{m} B_{m}^{-1} T_{q_A, q_B} B_{m - 1} A_{m - 1} \dots B_1 A_1 \ket{q_{0}}
		\overset
		{ ( \ref{eq:CSG Strong Winning Strategy Condition for Player 2 - III} ) }
		{ = }
		T_{q_A, q_B} \ket{q_{B}} = \ket{q_{A}} \ .
	\end{align}
	This last result contradicts (\ref{eq:CSG Strong Winning Strategy Condition for Player 2 - II}) and, therefore, player 2 does not have a strong winning strategy either.

	The existence of weak winning strategies for player 1 is easy to establish; the obvious choice $\sigma_{1} = (I, I, \dots, T_{q_A, q_B})$ is indeed such a strategy for player's 2 strategy $\sigma_{2} = (T_{q_0, q_B}, I, \dots, I)$. In an analogous manner we can show that player 2 has weak winning strategies too.
\end{proof}

This result is quite general and highlights the critical role of $S_{n}$ in quantum games. Furthermore, it allows us to immediately state two useful corollaries.

\begin{corollary}[Two classical players with the same action group] \label{crl:CSG Same Action Group for Classical Players}
	When two classical players use the same action group $S_{n}$ in a canonical game, neither of them has an advantage over the other, in the sense that neither possesses a strong winning strategy.
\end{corollary}

\begin{corollary}[Two quantum players with the same action group] \label{crl:CSG Same Action Group for Quantum Players}
	In a canonical game between two quantum players, if both players use the same action group that contains $S_{n}$, no player possesses a strong winning strategy.
\end{corollary}

These two corollaries prove that when both players have the same repertoire, which can be either classical, e.g., $S_{n}$, or quantum, e.g., a subgroup of $U(n)$, they are on an equal footing. This is a natural outcome of the overall symmetry in this case. The symmetry inherent in a canonical game is combined with the symmetry in the moves available to the players. If there is a moral in this, it would be that the pursuit of fairness, in the sense that both players have a chance to win, without any of them being certain to win, must involve symmetry.

\section{Results for asymmetric games} \label{sec:Results for Asymmetric Games}

The results of the previous section indicate that if we seek games where one players can surely win, we must turn to asymmetric settings. For this reason, in this section we shall consider three asymmetric scenarios:

\begin{enumerate}
	\item	The game is noncanonical, that is player 1, who makes the first move, also makes the last move, while both players have the same action group.
	\item	The game is canonical, but the action group $B$ of player 2 is a proper subgroup of the action group $A$ of player 1: $B < A$.
	\item	The game is noncanonical and, at the same time, the action group $B$ of player 2 is a proper subgroup of the action group $A$ of player 1.
\end{enumerate}

In the first of the above cases, the perceived advantage of player 1 is still not enough to guarantee a strong winning strategy. Hence, in this case too neither player is certain to win. This is proved in the next Theorem \ref{thr:NSG Noncanonical Game Same Action Group for Both Players}.

\begin{theorem}[Noncanonical game, same action group] \label{thr:NSG Noncanonical Game Same Action Group for Both Players}
	Consider a noncanonical game where player 1 makes the first and the last move and both players use the same action group. If the action group contains $S_{n}$, i.e., $S_{n} \leq A = B \leq U(n)$, no player has a strong winning strategy, but both players have a weak winning strategy.
\end{theorem}
\begin{proof}[Proof of Theorem \ref{thr:NSG Noncanonical Game Same Action Group for Both Players}]
	Suppose to the contrary that player 1 has a strong winning strategy $\sigma_{1}$ $=$ $(A_1,$ $A_2,$ $\dots,$ $A_{m},$ $A_{m + 1})$, where $A_1,$ $A_2,$ $\dots,$ $A_{m},$ $A_{m + 1}$ $\in$ $A$. This means that for every strategy $\sigma_{2}$ of player 2, player 1 will surely win.

	Consider the following two strategies of player 2: $\sigma_{2} = (B_1, B_2, \dots, B_{m - 1}, A_{m + 1}^{-1})$ and $\sigma'_{2}$ $=$ $(B_1,$ $B_2,$ $\dots,$ $B_{m - 1},$ $A_{m + 1}^{-1} T_{q_A, q_B})$, where $B_1, B_2, \dots, B_{m - 1}, A_{m + 1}^{-1}, A_{m + 1}^{-1} T_{q_A, q_B} \in B$ because $S_{n} \leq A = B \leq U(n)$. In view of Proposition \ref{prp:Condition for Strong Winning Strategy}, it must hold that
	\begin{align}
		A_{m + 1} A_{m + 1}^{-1} A_{m} \dots B_1 A_1 \ket{q_{0}} &= \ket{q_{A}} \label{eq:NSG Strong Winning Strategy Condition for Player 1 - I} \qquad \text{and}
		\\
		A_{m + 1} A_{m + 1}^{-1} T_{q_A, q_B} A_{m} \dots B_1 A_1 \ket{q_{0}} &= \ket{q_{A}} \label{eq:NSG Strong Winning Strategy Condition for Player 1 - II} \ .
	\end{align}
	Equation (\ref{eq:NSG Strong Winning Strategy Condition for Player 1 - I}) implies that
	\begin{align}
		A_m \dots B_1 A_1 \ket{q_{0}} = \ket{q_{A}} \label{eq:NSG Strong Winning Strategy Condition for Player 1 - III} \ .
	\end{align}
	Consequently,
	\begin{align} \label{eq:NSG Strong Winning Strategy Contradiction for Player 1}
		A_{m + 1} A_{m + 1}^{-1} T_{q_A, q_B} A_{m} \dots B_1 A_1 \ket{q_{0}}
		\overset
		{ ( \ref{eq:NSG Strong Winning Strategy Condition for Player 1 - III} ) }
		{ = }
		T_{q_A, q_B} \ket{q_{A}} = \ket{q_{B}} \ ,
	\end{align}
	which contradicts (\ref{eq:NSG Strong Winning Strategy Condition for Player 1 - II}). Thus, player 1 does not have a strong winning strategy.

	Similarly, in order to arrive at a contradiction, let us assume that player 2 has a strong winning strategy $\sigma_{2} = (B_1, B_2, \dots, B_{m - 1}, B_{m})$, where $B_1, B_2, \dots, B_m \in B$. Consider then the following two strategies of player 1: $\sigma_{1} = (A_1, A_2, \dots, A_{m}, I)$ and $\sigma'_{1} = (A_1, A_2, \dots, A_{m},$ $T_{q_A, q_B})$, where $A_1, A_2, \dots, A_{m}, I, T_{q_A, q_B} \in A$ because $A$ contains $S_{n}$. From Proposition \ref{prp:Condition for Strong Winning Strategy} we know that
	\begin{align}
		I B_{m} A_{m} \dots B_1 A_1 \ket{q_{0}} &= \ket{q_{B}} \label{eq:NSG Strong Winning Strategy Condition for Player 2 - I} \qquad \text{and}
		\\
		T_{q_A, q_B} B_{m} A_{m} \dots B_1 A_1 \ket{q_{0}} &= \ket{q_{B}} \label{eq:NSG Strong Winning Strategy Condition for Player 2 - II} \ .
	\end{align}
	Equation (\ref{eq:NSG Strong Winning Strategy Condition for Player 2 - I}) implies that
	\begin{align}
		B_{m} A_{m} \dots B_1 A_1 \ket{q_{0}} &= \ket{q_{B}} \ , \label{eq:NSG Strong Winning Strategy Condition for Player 2 - III}
	\end{align}
	which together with (\ref{eq:NSG Strong Winning Strategy Condition for Player 2 - II}) gives
	\begin{align} \label{eq:NSG Strong Winning Strategy Contradiction for Player 2}
		T_{q_A, q_B} B_{m} A_{m} \dots B_1 A_1 \ket{q_{0}}
		\overset
		{ ( \ref{eq:NSG Strong Winning Strategy Condition for Player 2 - III} ) }
		{ = }
		T_{q_A, q_B} \ket{q_{B}} = \ket{q_{A}} \ .
	\end{align}
	This last result contradicts (\ref{eq:NSG Strong Winning Strategy Condition for Player 2 - II}) and, therefore, player 2 does not have a strong winning strategy either.

	We can easily demonstrate the existence of weak winning strategies for both players using the same arguments as in Theorem \ref{thr:CSG Same action group for Both Players}.
\end{proof}

Let us emphasize that the above theorem holds when both players are classical and when both players are quantum.

Things are pretty much the same when the game is canonical, but the action group of one player is a proper subgroup of the action group of the other player, provided the smaller group contains $S_{n}$, as Theorem \ref{thr:CDG Canonical Game Different Action Group for Each Player} reveals.

\begin{theorem}[Canonical game, different action groups] \label{thr:CDG Canonical Game Different Action Group for Each Player}
	The following hold:
	\begin{itemize}
		\item	In a canonical game where the action group $A$ of player 1 contains $S_{n}$ and is a proper subgroup of the action group $B$ of player 2, i.e., $S_{n} \leq A < B \leq U(n)$, no player has a strong winning strategy.
		\item	Likewise, in a canonical game where the action group $B$ of player 2 contains $S_{n}$ and is a proper subgroup of the action group $A$ of player 1, i.e., $S_{n} \leq B < A \leq U(n)$, no player has a strong winning strategy.
	\end{itemize}
\end{theorem}
\begin{proof}[Proof of Theorem \ref{thr:CDG Canonical Game Different Action Group for Each Player}] \
	\begin{itemize}
		\item	Showing that player 1 does not have a strong winning strategy can be achieved by adapting the analogous arguments in Theorem \ref{thr:CSG Same action group for Both Players}.

		Therefore, we focus on showing that player 2 does not have a strong winning strategy either. To this end, let us suppose to the contrary that player 2 has a strong winning strategy $\sigma_{2} = (B_1, B_2, \dots, B_{m})$, where $B_1, B_2, \dots, B_{m} \in B$. This means that for every strategy $\sigma_{1}$ of player 1, player 2 will surely win. Consider the following two strategies of player 1: $\sigma_{1} = (T_{q_0, q_A}, A_2, \dots, A_{m})$ and $\sigma'_{2} = (T_{q_0, q_B}, A_2, \dots, A_{m})$, where $T_{q_0, q_A}, T_{q_0, q_B}, A_2, \dots, A_{m} \in A$ because $A$ contains $S_{n}$. Proposition \ref{prp:Condition for Strong Winning Strategy} entails that
		\begin{align}
			B_{m} A_{m} \dots B_1 T_{q_0, q_A} \ket{q_{0}} &= \ket{q_{B}} \label{eq:CDG Strong Winning Strategy Condition for Player 2 - I} \qquad \text{and}
			\\
			B_{m} A_{m} \dots B_1 T_{q_0, q_B} \ket{q_{0}} &= \ket{q_{B}} \label{eq:CDG Strong Winning Strategy Condition for Player 2 - II} \ .
		\end{align}
		If $C = B_{m} A_{m} \dots B_1$, then $C$ is a unitary operator, being the composition of unitary operators. Using $C$, equations (\ref{eq:CDG Strong Winning Strategy Condition for Player 2 - I}) and (\ref{eq:CDG Strong Winning Strategy Condition for Player 2 - II}) imply that
		\begin{align}
			C \ket{q_{A}} &= \ket{q_{B}} \label{eq:CDG Strong Winning Strategy Condition for Player 2 - III} \qquad \text{and}
			\\
			C \ket{q_{B}} &= \ket{q_{B}} \label{eq:CDG Strong Winning Strategy Condition for Player 2 - IV} \ .
		\end{align}
		This is impossible because $C$ is a unitary operator and, as such, sends different orthonormal vectors to different orthonormal vectors. Thus, player 2 does not have a strong winning strategy.

		\item	In this case it is straightforward to prove that player 2 does not have a strong winning strategy by suitably modifying analogous arguments in Theorem \ref{thr:CSG Same action group for Both Players}.

		Thus, we concentrate on showing that player 1 does not have a strong winning strategy either. To arrive at a contradiction, we suppose that player 1 has a strong winning strategy $\sigma_{1} = (A_1, A_2, \dots, A_{m})$, where $A_1, A_2, \dots, A_{m} \in A$. This means that for every strategy $\sigma_{2}$ of player 2, player 1 will surely win. Let us consider the following two strategies of player 2: $\sigma_{2} = (B_1, B_2, \dots, B_{m - 1}, I)$ and $\sigma'_{2} = (B_1, B_2, \dots, B_{m - 1},$ $T_{q_A, q_B})$, where $B_1, B_2, \dots, B_{m - 1}, I, T_{q_A, q_B} \in B$ because $B$ contains $S_{n}$. In view of Proposition \ref{prp:Condition for Strong Winning Strategy}, it must hold that
		\begin{align}
			I A_{m} \dots B_1 A_1 \ket{q_{0}} &= \ket{q_{A}} \label{eq:CDG Strong Winning Strategy Condition for Player 1 - I} \qquad \text{and}
			\\
			T_{q_A, q_B} A_{m} \dots B_1 A_1 \ket{q_{0}} &= \ket{q_{A}} \label{eq:CDG Strong Winning Strategy Condition for Player 1 - II} \ .
		\end{align}
		Equation (\ref{eq:CDG Strong Winning Strategy Condition for Player 1 - I}) implies that
		\begin{align}
			A_m \dots B_1 A_1 \ket{q_{0}} = \ket{q_{A}} \label{eq:CDG Strong Winning Strategy Condition for Player 1 - III} \ .
		\end{align}
		Consequently,
		\begin{align} \label{eq:CDG Strong Winning Strategy Contradiction for Player 1}
			T_{q_A, q_B} A_{m} \dots B_1 A_1 \ket{q_{0}}
			\overset
			{ ( \ref{eq:CDG Strong Winning Strategy Condition for Player 1 - III} ) }
			{ = }
			T_{q_A, q_B} \ket{q_{A}} = \ket{q_{B}} \ ,
		\end{align}
		which contradicts (\ref{eq:CDG Strong Winning Strategy Condition for Player 1 - II}). This concludes the proof that player 1 does not have a strong winning strategy.
	\end{itemize}
\end{proof}

We have thus arrived at an important realization. To ensure that one player has a strong winning strategy is not enough for her to play first and last if the other player has the same set of moves. Perhaps even more importantly, we have discovered that in a canonical game even if one player has access to a much larger action group, let's say the entire $U(n)$, she may still not win with probability $1.0$, provided the other player's action group contains $S_{n}$. Again, the hypothesis that the smaller group contains $S_{n}$ is absolutely critical. Unavoidably, to ensure the existence of a strong winning strategy for player 1 we must set the rules of the game so that she makes the first and the last move and, at the same time, her action group $A$ is sufficiently larger than the action group $B$ of player 2. This is what the following Theorem \ref{thr:NDG Noncanonical Game Different Action Group for Each Player} asserts. The proof is simple and is based on the concept of a state being unaffected by the action of all the elements of a group. In particular, we shall call a state $\ket{\psi}$ \emph{invariant} under the action of the group $B$, if for \emph{every} $B_i \in B$ it holds that $B_i \ket{\psi} = \ket{\psi}$. Of course, the identity operator $I$ fixes every state, but in this case the requirement is that all operators in $B$ fix the state $\ket{\psi}$.

\begin{theorem}[Noncanonical game and different action group] \label{thr:NDG Noncanonical Game Different Action Group for Each Player}
	Consider a noncanonical game where player 1 makes the first and the last move.
	\begin{itemize}
		\item	If the action group $A$ of player 1 contains $S_{n}$ as well as some action $U \in A$, which can drive the system from its initial state to a state $\ket{\psi}$ invariant under the action of $B$, then player 1 has a strong winning strategy.
		\item	If the action group $A$ of player 1 does not contain $S_{n}$, but contains some action $U \in A$ that can drive the system from its initial state to a state $\ket{\psi}$ invariant under the action of $B$, and her target state is the initial state, then player 1 has a strong winning strategy.
	\end{itemize}
\end{theorem}
\begin{proof}[Proof of Theorem \ref{thr:NDG Noncanonical Game Different Action Group for Each Player}] \
	\begin{itemize}
		\item	The fact that there exists $U \in A$ such that $U \ket{q_{0}} = \ket{\psi}$, implies that the strategy $\sigma_{1}$ $=$ $(U,$ $I,$ $\dots,$ $I,$ $T_{q_0, q_A} U^{- 1})$ is a strong winning strategy for player 1. Note that the composition $T_{q_0, q_A} U^{- 1}$ constitutes a move available to player 1 because $A$ contains $S_{n}$. It is easy to see that when player 2 makes her first move, the system is already at state $\ket{\psi}$, which, according to the hypothesis is invariant under the action of all operators in $B$. Thus, no matter what strategy player 2 employs the system will remain at state $\ket{\psi}$ until the last round of the game. Then during the last round of the game player 1 will drive the system to her target state $\ket{q_{A}}$ via the action $T_{q_0, q_A} U^{- 1}$.
		\item	In this case, although there exists $U \in A$ such that $U \ket{q_{0}} = \ket{\psi}$, the action group $A$ does not contain $S_{n}$ and, hence, transpositions cannot be used. The extra hypothesis in this case, which is that the target state of player 1 is $\ket{q_{0}}$, is critical. The strategy $\sigma_{1} = (U, I, \dots, I, U^{- 1})$ is a strong winning strategy for player 1. Again, it is easy to verify that when player 2 makes her first move, the system is already at state $\ket{\psi}$, which, according to the hypothesis is invariant under the action of all operators in $B$. Thus, no matter what strategy player 2 employs the system will remain at state $\ket{\psi}$ until the last round of the game. Then during the last round of the game player 1 will drive the system to back to its initial state $\ket{q_{0}}$, which is also her target state, via the action $U^{- 1}$.
	\end{itemize}
\end{proof}

At this point it may be expedient to clarify that the situation described in Theorem \ref{thr:NDG Noncanonical Game Different Action Group for Each Player} can only be realized when the action groups of the two players are different. It is straightforward to see that if $A = B$, whatever player 1 does to the system via a transformation $U$, can be undone by player 2 via $U^{- 1}$, which is available to her. In such a case, evidently, there exists no invariant state.

Theorem \ref{thr:NDG Noncanonical Game Different Action Group for Each Player} can help us to complete the picture formed by the previous three theorems. Theorems \ref{thr:CSG Same action group for Both Players} -- \ref{thr:CDG Canonical Game Different Action Group for Each Player} are somewhat negative by asserting the absence of a strong winning strategy. Theorem \ref{thr:NDG Noncanonical Game Different Action Group for Each Player} is more positive by giving conditions, or, a ``recipe'' for the design of games that ensure that one player will surely win. For example, this strategy is implicitly followed in the original PQ penny flip game of \cite{Meyer1999} and its clever generalizations to higher dimensions, like roulette, accomplished in \cite{Wang2000}, \cite{Ren2007}, and \cite{Salimi2009}. The following Table \ref{tbl:Conditions for the Existence of Strong Winning Strategy} summarizes the results we have obtained.

\hspace{0.5 cm}

\begin{table}[H]
	{\small
		\caption{This table presents the previous results in a compact form so that one can check the conditions under which a player will surely win. The abbreviation ``sws'' stands for strong winning strategy.}
		\label{tbl:Conditions for the Existence of Strong Winning Strategy}
		\renewcommand{\arraystretch}{1.5}
		\begin{center}
		\begin{tabular}{c !{\vrule width 1.25 pt} c|c}
			\Xhline{4 \arrayrulewidth}
			Type of game & sws for player 1 & sws for player 2
			\\
			\Xhline{3 \arrayrulewidth}
			Canonical \& same group (classical vs. classical) & No & No
			\\
			\hline
			Canonical \& same group (quantum vs. quantum) & No & No
			\\
			\hline
			Noncanonical \& same group (classical vs. classical) & No & No
			\\
			\hline
			Noncanonical \& same group (quantum vs. quantum) & No & No
			\\
			\hline
			Canonical \& $S_{n} \leq A < B \leq U(n)$ & No & No
			\\
			\hline
			Canonical \& $S_{n} \leq B < A \leq U(n)$ & No & No
			\\
			\hline
			Noncanonical \& $B < A$ \& $S_{n} \leq A$ \& $\ket{\psi}$ invariant & Yes & No
			\\
			\hline
			Noncanonical \& $B < A$ \& $\ket{\psi}$ invariant \& $\ket{q_0} = \ket{q_A}$ & Yes & No
			\\
			\Xhline{4 \arrayrulewidth}
		\end{tabular}
		\end{center}
		\renewcommand{\arraystretch}{1.0}
	}
\end{table}

\section{Discussion and examples} \label{sec:Discussion and Examples}

In this section we discuss the implications of the previous results and devise some toy games in the form of examples to see how the method outlined in Theorem \ref{thr:NDG Noncanonical Game Different Action Group for Each Player} can be used to guarantee that player 1 surely wins. The examples will employ finite automata, which, as mentioned in 
the Introduction, can be used as a powerful visual aid to enhance our understanding of the evolution of the state of the system and the progression of the game.

\begin{example} \label{xmp:A 4-Round Canonical Game with Different Groups}
	Let us consider a canonical $4$-round game on a $7$-dimensional quantum system. Let us further assume that $A = U(7)$, $B = S_{7}$, $\ket{q_0} = \ket{0}$, $\ket{q_A} = \ket{6}$, and $\ket{q_B} = \ket{0}$. Player 1 has the advantage of having a much larger group to choose from. $U(n)$ contains the $7$-dimensional quantum Fourier transform, which we denote by $F_{7}$. This is an extremely useful unitary transform with an elegant matrix representation (see \cite{Nielsen2010} for details). Its matrix representation is shown below, along with the result of its application to $\ket{0}$.
	\begin{align} \label{eq:The Quantum Fourier Transform $F_{7}$}
		F_{7}
		=
		\frac{1}{\sqrt{7}}
		\begin{bNiceMatrix}
			1 & 1 & 1 & 1 & 1 & 1 & 1 \\
			1 & \omega^{} & \omega^{2} & \omega^{3} & \omega^{4} & \omega^{5} & \omega^{6} \\
			1 & \omega^{2} & \omega^{4} & \omega^{6} & \omega^{} & \omega^{3} & \omega^{5} \\
			1 & \omega^{3} & \omega^{6} & \omega^{2} & \omega^{5} & \omega^{} & \omega^{4} \\
			1 & \omega^{4} & \omega^{} & \omega^{5} & \omega^{2} & \omega^{6} & \omega^{3} \\
			1 & \omega^{5} & \omega^{3} & \omega^{} & \omega^{6} & \omega^{4} & \omega^{2} \\
			1 & \omega^{6} & \omega^{5} & \omega^{4} & \omega^{3} & \omega^{2} & \omega^{}
		\end{bNiceMatrix}
		\quad \text{and} \quad
		\ket{\psi}
		=
		F_{7} \ket{0}
		=
		\frac{1}{\sqrt{7}}
		\begin{bNiceMatrix}
			1 \\
			1 \\
			1 \\
			1 \\
			1 \\
			1 \\
			1
		\end{bNiceMatrix} \ ,
	\end{align}
	where $\omega = e^{\frac{2 \pi i}{7}}$. Hence, player 1 can drive the system to the state $\ket{\psi} = \frac{1}{\sqrt{7}} \sum_{i = 0}^{6} \ket{i}$ during her first move by immediately applying $F_{7}$ to the initial state of the system. This particular state is invariant under $S_{n}$, so the system will remain at this state, irrespective of the second player's move.

	In the third round, player 1  makes her second and last move. If she decides to apply the composition of $T_{0, 6}$ and inverse action $F_{7}^{-1} = F_{7}^{\dagger}$, i.e., the action $T_{0, 6} F_{7}^{\dagger}$, then she will drive the system to her target state $\ket{6}$.

	However, it is quite possible that during her last move player 2 will use $T_{0, 6}$ to send the system back to the initial state $\ket{0}$, which happens to be her target state, and thus win the game.

	If, on the other hand, player 1 decides to use the identity operator $I$ as her last move, then player 2 will be unable to change the state of the system and, consequently, the system will be at state $\ket{\psi}$ prior to measurement. In this scenario player 1 fails again to win the game with probability $1.0$, having just $\frac{1}{7}$ probability to win.

	Naturally, all these are to be expected since Theorem \ref{thr:CDG Canonical Game Different Action Group for Each Player} asserts the nonexistence of a strong winning strategy. These two scenarios are graphically depicted in the automaton of Figure \ref{fig:Example 1 Automaton}.

	\begin{figure}[H]
		\centering
		\begin{tikzpicture}
			[
			scale = 1.3,
			state/.style =
			{
				circle, 
				draw, 
				minimum size = 1.0 cm,
				semithick,
			},
			every loop/.style = {min distance = 12 mm}
			]
			\def \angle {360/7}
			\node [ state, fill = GreenLighter2!20, initial, initial text = { initial state }, initial where = right, initial distance = 0.5 cm, accepting ] (q0)
			at ( {3 * cos(0 * \angle)}, {3 * sin(0 * \angle)} ) { $ \ket{0} $ };
			\node (t2) at ( 4.75, -0.5 ) { \& Player's 2 target };
			\node [ state ] (q1) at ( {3 * cos(1 * \angle)}, {3 * sin(1 * \angle)} ) { $ \ket{1} $ };
			\node [ state ] (q2) at ( {3 * cos(2 * \angle)}, {3 * sin(2 * \angle)} ) { $ \ket{2} $ };
			\node [ state ] (q3) at ( {3 * cos(3 * \angle)}, {3 * sin(3 * \angle)} ) { $ \ket{3} $ };
			\node [ state ] (q4) at ( {3 * cos(4 * \angle)}, {3 * sin(4 * \angle)} ) { $ \ket{4} $ };
			\node [ state ] (q5) at ( {3 * cos(5 * \angle)}, {3 * sin(5 * \angle)} ) { $ \ket{5} $ };
			\node [ state, fill = RedPurple!15 ] (q6) at ( {3 * cos(6 * \angle)}, {3 * sin(6 * \angle)} ) { $ \ket{6} $ };
			\node (t1) at ( 3.5, -2.375 ) { Player's 1 target };
			\node [ state, fill = WordBlueDark!50 ] (psi) at ( 0.0, 0.0 ) { $ \ket{\psi} $ };
			\path[->]
			(q0) edge [bend right = 25] node [ above ] { $F_{7}$ } (psi)
			(q0) edge [bend right = 25] node [ left ] { $T_{0, 6}$ } (q6)
			(q6) edge [bend right = 25] node [ right ] { $T_{0, 6}$ } (q0)
			(psi) edge [bend right = 25] node [ above ] { $F_{7}^{\dagger}$ } (q0)
			(psi) edge [out = 150, in = -150, loop] node [left] { $I, T_{i, j}$ } ();
		\end{tikzpicture}
		\caption{This automaton captures the evolution of the game studied in this example according to the two possible strategies of player 1 outlined above.}
		\label{fig:Example 1 Automaton}
	\end{figure}
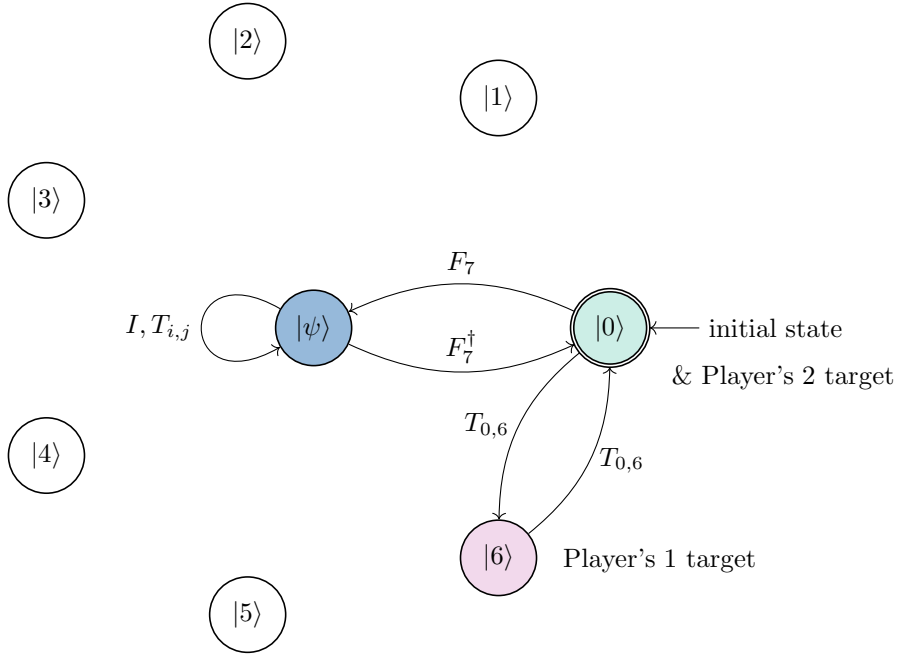
\end{example}

\begin{example} \label{xmp:A 5-Round Noncanonical Game with Different Groups}
	In the previous game player 1 cannot surely win. Therefore, in order to guarantee that player 1 wins with probability $1.0$, the rules of the game must change. In particular, the asymmetry and the disadvantage for player 2 must increase. To this end, let us consider a noncanonical $5$-round game on the same $7$-dimensional quantum system, in which, again, we assume that $A = U(7)$, $B = S_{7}$, $\ket{q_0} = \ket{0}$, $\ket{q_A} = \ket{6}$, and $\ket{q_B} = \ket{0}$.

	The obvious difference in the current game is that now player 1 has a double advantage:

	\begin{enumerate}
		\item	she makes both the first and the last move, and
		\item	she has a much larger group to choose from.
	\end{enumerate}

	The same quantum Fourier transform $F_{7}$ shown in equation (\ref{eq:The Quantum Fourier Transform $F_{7}$}) of the previous example, can now pave the way to much better results. 

	Specifically, the strategy $\sigma_{1} = ( F_{7}, I, \dots, I, T_{0, 6} F_{7}^{\dagger} )$ is a strong winning strategy for player 1. This strategy of player 1 drives the system to the state $\ket{\psi} = \frac{1}{\sqrt{7}} \sum_{i = 0}^{6} \ket{i}$, which is invariant under $S_{n}$. Hence, the system will remain at this state, irrespective of the second player's move.

	The critical realization here is that the system, in contrast to the previous example, is still at state $\ket{\psi}$ after player 2 has completed her moves.

	The last move is now by player 1, who can act using the operator $T_{0, 6} F_{7}^{\dagger}$ to drive the system to her target state $\ket{6}$. In this game player 1 will surely win. This winning scenario is captured by the automaton of Figure \ref{fig:Example 2 Automaton}.

	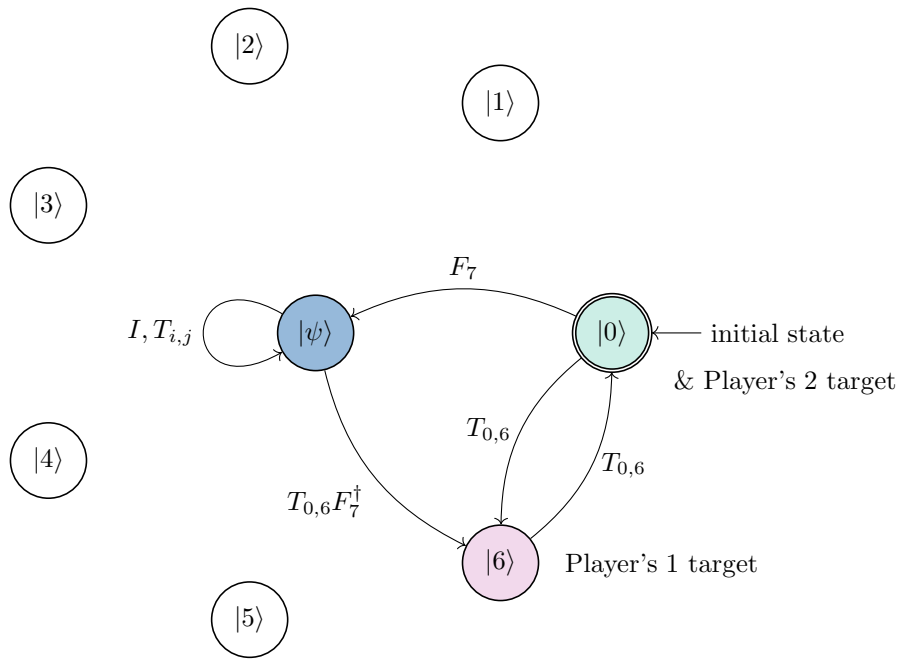
\begin{figure}[H]
		\centering
		\begin{tikzpicture}
			[
			scale = 1.3,
			state/.style =
			{
				circle, 
				draw, 
				minimum size = 1.0 cm,
				semithick,
			},
			every loop/.style = {min distance = 12 mm}
			]
			\def \angle {360/7}
			\node [ state, fill = GreenLighter2!20, initial, initial text = { initial state }, initial where = right, initial distance = 0.5 cm, accepting ] (q0)
			at ( {3 * cos(0 * \angle)}, {3 * sin(0 * \angle)} ) { $ \ket{0} $ };
			\node (t2) at ( 4.75, -0.5 ) { \& Player's 2 target };
			\node [ state ] (q1) at ( {3 * cos(1 * \angle)}, {3 * sin(1 * \angle)} ) { $ \ket{1} $ };
			\node [ state ] (q2) at ( {3 * cos(2 * \angle)}, {3 * sin(2 * \angle)} ) { $ \ket{2} $ };
			\node [ state ] (q3) at ( {3 * cos(3 * \angle)}, {3 * sin(3 * \angle)} ) { $ \ket{3} $ };
			\node [ state ] (q4) at ( {3 * cos(4 * \angle)}, {3 * sin(4 * \angle)} ) { $ \ket{4} $ };
			\node [ state ] (q5) at ( {3 * cos(5 * \angle)}, {3 * sin(5 * \angle)} ) { $ \ket{5} $ };
			\node [ state, fill = RedPurple!15 ] (q6) at ( {3 * cos(6 * \angle)}, {3 * sin(6 * \angle)} ) { $ \ket{6} $ };
			\node (t1) at ( 3.5, -2.375 ) { Player's 1 target };
			\node [ state, fill = WordBlueDark!50 ] (psi) at ( 0.0, 0.0 ) { $ \ket{\psi} $ };
			\path[->]
			(q0) edge [bend right = 25] node [ above ] { $F_{7}$ } (psi)
			(q0) edge [bend right = 25] node [ left ] { $T_{0, 6}$ } (q6)
			(q6) edge [bend right = 25] node [ right ] { $T_{0, 6}$ } (q0)
			(psi) edge [bend right = 25] node [ below left ] { $T_{0, 6} F_{7}^{\dagger}$ } (q6)
			(psi) edge [out = 150, in = -150, loop] node [left] { $I, T_{i, j}$ } ();
		\end{tikzpicture}
		\caption{This automaton displays the strong strong winning strategy of player 1 for the game studied in this example.}
		\label{fig:Example 2 Automaton}
	\end{figure}

\end{example}

The results and the examples above should help drive through the main point of this study, which is that the \emph{specific rules of a game are extremely important}! A slight variation can have profound consequences. It is the combination of two factors:

\begin{enumerate}
	\item	the precise sequence of moves as dictated by the rules of the game, and
	\item	the sets of admissible actions the players draw their moves from,
\end{enumerate}

that determines who wins. Quantum strategies do not a priori prevail over classical strategies. With some ingenuity one may design a game where the quantum player does not have a winning strategy against the classical player. See such an illuminating example and further comments in \cite{Anand2015}. The determining factor are the rules of the game. By carefully designing the rules of the game, the advantage of either player can be established. Alternatively, the fairness of the game can also be guaranteed. Another immediate conclusion that can be drawn is that although it is indeed possible to have quantum versions of classical games that exhibit different behavior or even entirely new traits, caution must exercised in designing the rules that will make this possible.

\section{Conclusions} \label{sec:Conclusions}

Quantum games can be used in a plethora of ``serious'' situations. For instance rolling quantum dice can be used to implement secure cryptographic protocols. In this paper we studied sequential quantum games under the assumption that the moves of the players are drawn from groups and not plain unstructured sets. The extra group structure makes possible to easily derive very general results as the those in Sections \ref{sec:Results for Canonical Games} and \ref{sec:Results for Asymmetric Games}, which, to the best of our knowledge, are stated in this generality for the first time. Although it seems a bit restrictive to focus on groups, it is consistent with the quantumness of these games. After all, quantum actions come from unitary groups.

Of course, the possibilities for quantum games are virtually endless. Many important and critical issues have to be addressed, such as to what extent the players communicate, what type of information is known to the players, if and how precisely entanglement comes into play. Indeed, the degree of knowledge of the players about the system is critical in the formation of strategies as we know from classical game theory. One of course will ask how can this be, without measuring the system? It is of course virtually impossible, but perhaps we can imagine a third party spying on the players and acquiring partial information about their strategies. This is, we believe, a very intriguing topic which we intend to pursuit in a future work.

\bibliographystyle{ieeetr}
\bibliography{ConditionsSurelyWinSequentialQuantumGames}

\end{document}